\RequirePackage[hyphens]{url}
\documentclass[sigconf]{acmart}

\usepackage{booktabs} % For formal tables

\usepackage{graphicx}  %Required

\usepackage{amsmath, amsthm, amssymb}

\newtheorem{proposition}{Proposition}

\newtheorem{theorem}{Theorem}

\usepackage{graphicx}
\usepackage[utf8]{inputenc} % allow utf-8 input
\usepackage[T1]{fontenc}    % use 8-bit T1 fonts
\usepackage{amsfonts}       % blackboard math symbols
\usepackage{nicefrac}       % compact symbols for 1/2, etc.
\usepackage{xcolor}

\usepackage[lined, boxed, ruled, commentsnumbered, noend]{algorithm2e}
\usepackage{subcaption}
\usepackage{xspace}

\usepackage{tikz}
\usetikzlibrary{matrix,arrows,decorations.pathmorphing}
\usetikzlibrary{positioning}

\newcommand{\E}{\ensuremath{\mathbb{E}}}

\usepackage{xfrac}

\newcommand{\fake}{\ensuremath{{f}}}
\newcommand{\nofake}{\ensuremath{{\bar{f}}}}
\newcommand{\UserPrior}{\ensuremath{\Theta}}
\newcommand{\userPrior}{\ensuremath{\theta}}
\newcommand{\newsPrior}{\ensuremath{\omega}}
\newcommand{\userHistory}{\ensuremath{\mathcal{D}}}

\newcommand{\val}{\ensuremath{\textnormal{val}}}

\newcommand{\oracle}{\textsc{Oracle}\xspace}
\newcommand{\algo}{\textsc{Detective}\xspace}
\newcommand{\algoStar}{\textsc{Opt}\xspace}
\newcommand{\opt}{\textsc{Opt}\xspace}
\newcommand{\random}{\textsc{Random}\xspace}
\newcommand{\nolearn}{\textsc{No-Learn}\xspace}
\newcommand{\fixedCM}{\textsc{Fixed-CM}\xspace}
\newcommand{\generalalgo}{\textsc{Algo}\xspace}
\newcommand{\topx}{\textsc{TopX}\xspace}
\newcommand{\Time}{\ensuremath{T}}

\begin{document}

\copyrightyear{2018}
\acmYear{2018} 
\setcopyright{iw3c2w3}
\acmConference[WWW '18 Companion]{The 2018 Web Conference Companion}{April 23--27, 2018}{Lyon, France}
\acmBooktitle{WWW '18 Companion: The 2018 Web Conference Companion, April 23--27, 2018, Lyon, France}
\acmPrice{}
\acmDOI{10.1145/3184558.3188722}
\acmISBN{978-1-4503-5640-4/18/04}

\fancyhead{}

% \title{Leveraging Crowd Signals for Fake News Detection}
\title{Fake News Detection in Social Networks via Crowd Signals}
% \title{Detecting Fake News in Social Networks via Crowdsourcing}
%\titlenote{Produces the permission block, and
%  copyright information}
%\subtitle{Extended Abstract}
%\subtitlenote{The full version of the author's guide is available as
%  \texttt{acmart.pdf} document}

%\author{Anonymous}
%%\authornote{Work performed while at ETH Zurich.}
%\affiliation{%
%  \institution{Affiliation}
%%  \city{Cambridge}
%%  \country{United Kingdom}
%}

\author{Sebastian Tschiatschek}
\authornote{Work performed while at ETH Zurich.}
\affiliation{%
  \institution{Microsoft Research}
  \city{Cambridge}
  \country{United Kingdom}
}
\email{setschia@microsoft.com}

\author{Adish Singla}
\affiliation{%
  \institution{MPI-SWS}
  \city{Saarbr{\"u}cken}
  \country{Germany}
}
\email{adishs@mpi-sws.org}

\author{Manuel Gomez Rodriguez}
\affiliation{%
  \institution{MPI-SWS}
  \city{Kaiserslautern}
  \country{Germany}}
\email{manuelgr@mpi-sws.org}

\author{Arpit Merchant}
\affiliation{%
  \institution{IIIT-H}
  \city{Hyderabad}
  \country{India}
}
\email{arpitdm@gmail.com}

\author{Andreas Krause}
\affiliation{%
 \institution{ETH Zurich}
 \city{Zurich}
 \country{Switzerland}}
\email{krausea@ethz.ch}

\renewcommand{\shortauthors}{S. Tschiatschek et al.}

%!TEX root = main.tex

%%%%%%%%%%%%%%%%%%%%%%%%%%%%%%%%%%%%%%%%%%%%%%%%%%%%%%%%%%
%%%%%%%%%%%%%%%%%%%%%%%%%%%%%%%%%%%%%%%%%%%%%%%%%%%%%%%%% ABSTRACT
\begin{abstract}
Our work considers leveraging crowd signals for detecting fake news and is motivated by tools recently introduced by Facebook that enable users to flag fake news. By aggregating users' flags, our goal is to select a small subset of news every day, send them to an expert (e.g., via a third-party fact-checking organization), and stop the spread of news identified as fake by an expert. The main objective of our work is to minimize \emph{the spread of misinformation} by stopping the propagation of fake news in the network. It is especially challenging to achieve this objective as it requires detecting fake news with high-confidence as quickly as possible. We show that in order to leverage users' flags efficiently, it is crucial to learn about users' flagging accuracy. We develop a novel algorithm, \algo, that performs Bayesian inference for detecting fake news and jointly learns about users' flagging accuracy over time.
% based on history of users' flagging activity and expert's labels obtained.
Our algorithm employs posterior sampling to actively trade off exploitation (selecting news that maximize the objective value at a given epoch) and exploration (selecting news that maximize the value of information towards learning about users' flagging accuracy). We demonstrate the effectiveness of our approach via extensive experiments and show the power of leveraging community signals for fake news detection.
% towards leveraging crowd signals for detecting fake news.
%on a publicly available Facebook dataset.
\end{abstract}

\maketitle

%%%%%%%%%%%%%%%%%%%%%%%%%%%%%%%%%%%%%%%%%%%%%%%%%%%%%%%%%%%
%%%%%%%%%%%%%%%%%%%%%%%%%%%%%%%%%%%%%%%%%%%%%%%%%%%%%%%%%%% INTRODUCTION
%!TEX root = main.tex

\section{Introduction}
Fake news (a.k.a. hoaxes, rumors, etc.) and the spread of misinformation have dominated the news cycle since the US presidential election (2016). Social media sites and online social networks, for example Facebook and Twitter, have faced scrutiny for being unable to curb the spread of fake news. There are various motivations for generating and spreading fake news, for instance, making political gains, harming the reputation of businesses, as clickbait for increasing advertising revenue, and for seeking attention\footnote{Snopes compiles a list of top 50 fake news stories: \texttt{http://www.snopes.com/50-hottest-urban-legends/}}. As a concrete example, Starbucks recently fell victim to fake news with a hoax advertisement claiming that the coffee chain would give free coffee to undocumented immigrants\footnote{\url{http://uk.businessinsider.com/fake-news-starbucks-free-coffee-to-undocumented-immigrants-2017-8}}. While Starbucks raced to deny this claim by responding to individual users on social media, the lightening speed of the spread of this hoax news in online social media highlighted the seriousness of the problem and the critical need to develop new techniques to tackle this challenge. To this end, Facebook has recently announced a series of efforts towards tackling this challenge \cite{facebook-news-feed,facebook-germany}.

\noindent{\bfseries Detection via expert's verification.} 
Fake news and misinformation have historically been used as tools for making political or business gains \cite{ewen1998pr}. However, traditional approaches based on verification by human editors and expert journalists do not scale to the volume of news content that is generated in online social networks. In fact, it is this volume as well as the lightening speed of spread in these networks that makes this problem challenging and requires us to develop new computational techniques. We note that such computational techniques would typically complement, and not replace, the expert verification process---even if a news is detected as fake, some sort of expert verification is needed before one would actually block it.  This has given rise to a number of third-party fact-checking organizations such as Snopes\footnote{\url{http://www.snopes.com/}} and Factcheck.org\footnote{\url{http://factcheck.org/}} as well as a code of principles~\cite{ifcn} that should be followed by these organizations.
%\cite{snopes}, \cite{factcheck}
%approaches are not supposed to replace the expert's verification process---
%of experti
%
%tackle this problem 
%
%The challenge of dealing with fake news and misinformation is not just an artifact of social media
%
% is to hire human editors and expert journalists to review the news content before it gets posted. However, given the this has been the traditional way of fact checking
%%%%
%Late last year, the company announced that it would be working to identify these stories to users, with the help of five independent fact-checking organisations: Snopes, Politifact, FactCheck.org, the Associated Press and ABC News.
%However, hiring a large number of experts to deal with the volume of online news content is extremely expensive and impossible to scale.

\noindent{\bfseries Detection using computational methods.} There has been a recent surge in interest towards developing computational methods for detecting fake news (cf., \cite{conroy2015automatic} for a survey)---we provide a more detailed overview of these methods in the Related Work section. These methods are typically based on building predictive models to classify whether a news is fake or not via using a combination of features related to news content, source reliability, and network structure. One of the major challenges in training such predictive models is the limited availability of corpora and the subjectivity of labelling news as fake~\cite{wang2017liar,rubin2015deception}. Furthermore, it is difficult to design methods based on estimating source reliability and network structure as the number of users who act as sources is diverse and gigantic (e.g., over one billion users on Facebook); and the sources of fake news could be normal users who unintentionally share a news story without realizing that the news is fake. A surge of interest in the problem and in overcoming these technical challenges has led to the establishment of a volunteering based association---FakeNewsChallenge\footnote{\texttt{http://www.fakenewschallenge.org/}}---comprising over 100 volunteers and 70 teams which organizes machine learning competitions related to the problem of detecting fake news. 

%; there is also a large body of related work on rumor detection and information credibility evaluation that are applicable to the problem of detecting fake news
%%%%%%%%
%are also difficult to design given the diverse and large user pool in online social networks who could be potential sources; a lot of times, fake news  
%
%are limited corpora and abmibutiy 
%
%methods
%
%
%challenges as corpora and problem not well defined \cite{rubin2015deception}. Recent challenge . Sources are changing and large. In short, these techniques have generally been unsuccessfuly to stop the recent issues, and large false positives requires experts validation. Another approach is to use algorithmic solutions based on machine learning to detect fake news. However, so far, they have turned out to be inadequate in fixing this problem.
%%\cite{wang2017liar}: "Liar, Liar Pants on Fire": {A} New Benchmark Dataset for Fake News Detection

%%%%%%%%%%%%%%%%%%%%%%%%%%%%%%%%%%%%%%%%%%%%%%%%%%%%%%%%%%%%%%%%%%%%%%%
%%%%%%%%%%%%%%%%%%%%%%%%%%%%%%%%%%%%%%%%%%%%%%%%%%%%%%%%%%%%%%%%%%%%%%%
%%%%%%%%%%%%%%%%%%%%%%%%%%%%%%%%%%%%%%%%%%%%%%%%%%%%%%%%%%%%%%%%%%%%%%%
\subsection{Leveraging users' flags.}
%\noindent{\bfseries Leveraging users' flags.}
%Hybrid human-AI methods 
Given the limitation of the current state-of-the-art computational methods, an alternate approach is to develop hybrid human-AI methods via engaging users of online social networks by enabling them to report fake news. In fact, Facebook has recently taken steps towards this end by launching a fake news reporting tool in Germany \cite{facebook-germany}, as shown in Figure~\ref{fig.fakenews-tool}. The idea of this tool is that as news propagates through the network, users can flag the news as fake. 
\begin{figure}[!h]
	\centering
    \includegraphics[width=0.48\textwidth]{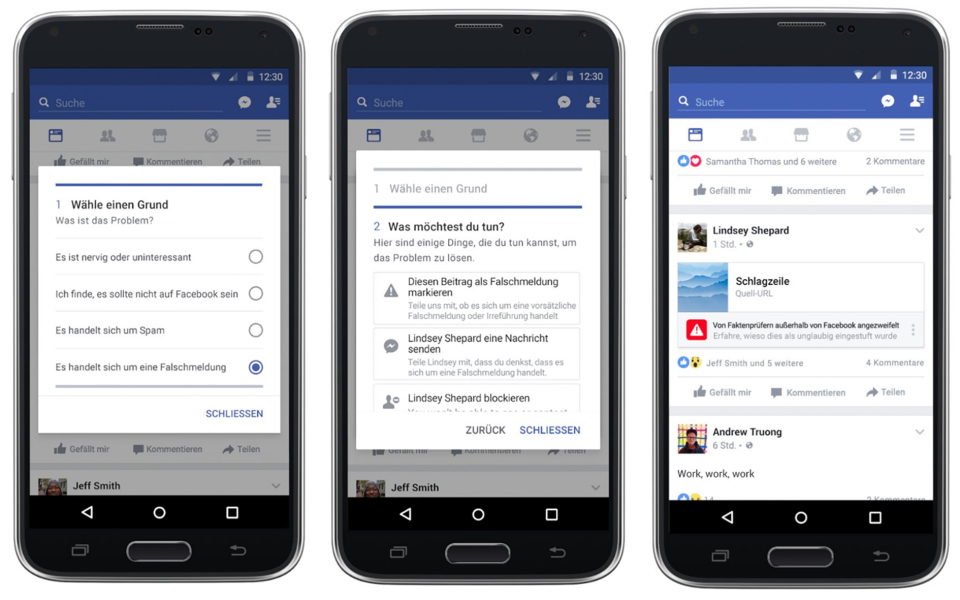}
	\caption{Facebook has launched tools in Germany to report fake news. Image source: \cite{facebook-germany}.}
    \label{fig.fakenews-tool}
\end{figure}

\noindent As proposed by~\citet{facebook-germany}, the aggregated users' flags as well as well as other available signals can be used to identify a set of news which potentially is fake. These news can then be sent to an expert for review via a third-party fact-checking organization. If an expert labels the news as fake, it could be removed from the network or marked as disputed making it appear lower in news-feed ranking.
The contemporary work by \citet{curb18wsdm} explored the idea of detecting fake news via leveraging users' flagging activity by using the framework of marked temporal point processes. We highlight the key differences of their approach to ours in the next section.

\subsection{Our Contributions} 

%\noindent{\bfseries Our contributions.} 
In this paper, we develop algorithmic tools to effectively utilize the power of the crowd (flagging activity of users) to detect fake news. 
%Below, we provide an overview of our approach and summarize the main contributions.
Given a set of news, our goal is to select a small subset of $k$ news, send them to an expert for review, and then block the news which are labeled as fake by the expert. We formalize our objective as to minimize \emph{the spread of misinformation}, i.e., how many users end up seeing a fake news before it is blocked. 
%Thus, in order to achieve a high objective, it is critical to not only select news which could be fake with high-confidence, but also to select them as early as possible. 
%
%We argue, and demonstrate via our experiments, that in order to leverage the users' flags, it is critical to learn about users' accuracy---any algorithm that aggregates users flags via treating all users same is prone to spam and adversarial attacks.
We design our algorithm \algo, which implements a Bayesian approach for learning about users' accuracies over time as well as for performing inference to find which news are fake with high confidence. 
% In general, designing an optimal policy to decide which $k$ news should be selected at every epoch $t$ is challenging, especially because policies' current actions (the news it selects at time $t$) influence the objective value, available actions, and knowledge about users in the future.
%In order to deal with this challenge, we divide the global optimization problem in to simpler local optimization problems per epoch by adding an extra exploration term that encourages selecting news that maximize the value of information towards learning about users' accuracy.
%We then design our algorithm, \algo, which executes per epoch $t \in [T]$ separately to select the $k$ news. By exploiting the submodular property of the local optimization problem at any given epoch $t$, algorithm \algo selects news which are near-optimal w.r.t. the optimal set at $t$.
%
In short, our main contributions include:
\begin{itemize}
  \item We formalize the problem of leveraging users' flagging activity for detection of fake news. We showcase the need to learn about users' accuracy in order to effectively leverage their flags in a robust way.
  \item We develop a tractable Bayesian algorithm, \algo, that actively trades off between exploitation (selecting news that directly maximize the objective value) and exploration (selecting news that helps towards learning about users' flagging accuracy).
  \item We perform extensive experiments using a publicly available Facebook dataset to demonstrate the effectiveness of our approach. We plan to make the code publicly available so that other researchers can build upon our techniques for this important and timely problem of detecting fake news.
\end{itemize}

\section{Related Work}

{\bfseries Contemporary results.}
\citet{curb18wsdm} explored the idea of detecting fake news via leveraging users' flagging activity. In particular, they introduce a flexible representation of the above problem using the framework of marked temporal point processes. They develop an algorithm, \textsc{Curb}\xspace, to select which news to send for fact-checking via solving a novel stochastic optimal control problem.  The key technical differences of the approach by \citet{curb18wsdm} to ours are:
(1) we learn about the flagging accuracy of individual users in an online setting; in contrast, they consider all users to be equally reliable and estimate the flagging accuracy of the population of users from historical data;
(2) our algorithms are agnostic to the actual propagation dynamics of news in the network; they model the actual propagation dynamics as a continuous-time dynamical system with jumps and arrive at an algorithm by casting the problem as an optimal control problem; and
(3) we use discrete epochs with a fixed budget per epoch (\emph{i.e.}, the number of news that can be sent to an expert for reviewing); they use continuous time and consider an overall budget for their algorithm.

\noindent {\bfseries Computational methods for detecting fake news.}
There is a large body of related work on rumor detection and information credibility evaluation (with a more recent focus on fake news detection) that are applicable to the problem of detecting fake news. These methods are typically based on building predictive models to classify whether a news is fake. At a high-level level, we can categorize these methods as follows: (i) based on features using news content via natural language processing techniques \cite{zhao2015enquiring,wei2017learning,volkova2017separating,gupta2014tweetcred}; (ii) via learning models of source reliability and trustworthiness  \cite{li2015discovery,tabibian2017distilling,lumezanu2012bias}; (iii) by analyzing the network structure over which a news propagated \cite{ciampaglia2015computational}; and (iv) based on a combination of the above-mentioned features, i.e., linguistic, source, and network structure \cite{kwon2017rumor,castillo2011information,wu2016information,kumar2016disinformation}. As we pointed out in the Introduction, there are several key challenges in building accurate predictive models for identifying fake news including limited availability of corpora, subjectivity in ground truth labels, and huge variability in the sources who generate fake news (often constituting users who do it unintentionally). In short, these methods alone have so far proven to be unsuccessful in tackling the challenge of detecting fake news.

\noindent{\bfseries Leveraging crowd signals for web applications.}
Crowdsourcing has been used in both industrial applications and for research studies in the context of different applications related to web security. For instance, \cite{moore2008evaluating} and \cite{chia2011re} have evaluated the potential of leveraging the wisdom of crowds for assessing phishing websites and web security. Their studies show a high variability among users---(i) the participation rates of users follows a power-law distribution, and (ii) the accuracy of users' reports vary, and users with more experience tend to have higher accuracy. The authors also discuss the potential of voting fraud when using users' reports for security related applications. \citet{wang2013social} performed a crowdsourcing study on Amazon's Mechanical Turk for the task of sybil detection in online social networks. Their studies show that there is a huge variability among crowd users in terms of their reporting accuracies that needs to be taken into account for building a practical system. \citet{zheleva2008trusting,chen2015trusms} present a system similar to that of ours for the task of filtering email spam and SMS spam, respectively. The authors discuss a users' reputation system whereby reliable users (based on history) can be weighted more when aggregating the reports. However, their work assumes that users' reputation/reliability is known to the system, whereas the focus of our paper is on learning users' reputation over time. 
%In a recent work, \cite{sethi2017crowdsourcing} presents an overview of a system architecture to detect fake news via engaging users and leveraging community systems, however there are no formal objective or algorithmic techniques to be used in such a system.
\citet{freeman2017can} discusses the limitations of leveraging user feedback for fake account detection in online social networks---via data-driven studies using Linkedin data, the authors show that there is only a small number of skilled users (who have good accuracy that persists over time) for detecting fake accounts.
%However, authors point out that by providing appropriate incentives or targeting users' to participate actively could could increase the number of skilled users.

\noindent{\bfseries Crowdsourcing with expert validation}
On a technical side, our approach can be seen as that of a semi-supervised crowdsourcing technique where users' answers can be validated via an external expert. \citet{hung2015minimizing,liu2017improving} present probabilistic models to select specific news instances to be labeled by experts that would maximize the reduction in uncertainty about users' accuracy.
%In our algorithm, \algo, the exploration component of the objective aims to select news with similar objective, however our overall objective is much different from that studied in the above-mentioned papers.
With a similar flavor to ours, \citet{zhao2012bayesian} presents a Bayesian approach to aggregate information from multiple users, and then jointly infer users' reliability as well as ground truth labels. Similar to our approach, they model users' accuracy via two separate parameters for false positive and false negative rates. However, their approach is studied in an unsupervised setting where no expert validation (ground truth labels) are available.

\renewcommand{\algorithmcfname}{Protocol}
%\begin{algorithm2e}[t!]
%\noindent\begin{minipage}{\textwidth}
%\renewcommand\footnoterule{}                  %% This line should come here.
\begin{algorithm*}[t!]
\nl {\bf Input}: {social network graph $G = (U, E)$; labeling budget per epoch $k$}.\\
\nl {\bf Initialize}: {active news $A^0 = \{\}$ (i.e., news for which expert's label is not acquired yet)}.\\
%;\\\qquad \qquad \ user history $D^t(u) = \{\} \ \forall u \in U$}.\\
	\nl \ForEach {$t = 1, 2, \ldots, \Time$}{
	    \tcc*[h]{\textcolor{blue}{At the beginning of epoch $t$}}\\
		\nl News $X^t$ are generated with $o_x \in U$ as the origin/source of $x \in X^t$.\\ \label{interaction.generation}
		\nl Update the set of active news as $A^t = A^{t-1} \cup X^t$. $\forall x \in X^t$, do the following: \label{interaction.active} \\  		
			\Indp			
			\nl Initialize users exposed to the news $x$ as $\pi^{t-1}(x) = \{\}$. \\
%\{o_x\}$ where $o_x \in U$ is the origin/source of $x$ \label{interaction.source}.			
			\nl Initialize users who flagged the news $x$ as $\psi^{t-1}(x) = \{\}$.\\
			%\nl Add $x$ to the set of active items as  $A^t = A^{t-1} \cup \{x\}$. \label{interaction.activeset}\\
			\Indm
	    \tcc*[h]{\textcolor{blue}{During the epoch $t$}} \\
		\nl News $A^t$ continue to propagate in the network. $\forall a \in A^t$, do the following: \\
			\Indp
			\nl News $a$ propagates to more users $u^t(a) \subseteq U \setminus \pi^{t-1}(a)$; i.e., $\pi^t(a) = \pi^{t-1}(a) \cup u^t(a)$. \label{interaction.spread}\\
			\nl News $a$ is flagged as fake by users $l^t(a) \subseteq u^t(a)$; i.e., $\psi^t(a) = \psi^{t-1}(a) \cup l^t(a)$. \label{interaction.flag}\\			
			\Indm
	    \tcc*[h]{\textcolor{blue}{At the end of epoch $t$}} \\
		\nl Algorithm \generalalgo selects a subset  $S^t \subseteq A^t$ of up to size $k$ to get expert's labels given by $y^*(s) \in \{f, \bar{f}\} \ \forall \ s \in S^t$.  \label{interaction.select}\\
			\Indp
			\nl Block the fake news, i.e., $\forall s \in S^t \text{ s.t. } y^*(s) = f$, remove $s$ from the network. \label{interaction.select_block}\\
			\nl Update the set of active news as $A^t = A^t \setminus S^t$ \\			       
			{{\footnotesize{Note that news $s \in S^t \text{ s.t. } y^*(s) = \bar{f}$ remain in the network, continue to propagate, and being flagged by users}}} \label{interaction.select_updateactive}\\
%			\nl Update user history, i.e., $\forall s \in S^t, u \in \pi^t(s) \setminus \{o(s)\}$, do the following:\\
			\Indm
%			\nl \generalalgo updates users' history, i.e., $\forall s \in S^t, u \in \pi^t(s) \setminus \{o(s)\}$, do the following: \label{interaction.update_history}\\
%			\Indp
%			\If {$(u \in \psi^t(s))$}{
%				\nl $D^t(u) = D^{t-1}(u) \cup \{(s, Y(s), f)\}$ 
%			}
%			\Else {
%				\nl $D^t(u) = D^{t-1}(u) \cup \{(s, Y(s), \bar{f})\}$
%			}
%			\Indm
  }			
  \caption{High-level specification of our model}
  \label{interaction}
\end{algorithm*}
%\end{minipage}
%%%%%%%%%%%%%%%%%%%%%%%%%%%%%%%%%%%%%%%%%%%%%%%%%%%%%%%%%%%%%%%%%%%%%%%%%%%%%%
%%%%%%%%%%%%%%%%%%%%%%%%%%%%%%%%%%%%%%%%%%%%%%%%%%%%%%%%%%%%%%%%%%%%%%%%%%%%%%

\section{The Model}
We provide a high-level specification of our model in Protocol~\ref{interaction}. There is an underlying social network denoted as $G=(U,E)$ where $U$ is the set of users in the network. We divide the execution into different epochs denoted as $t = 1, 2, \ldots, T$, where each epoch could denote a time window, for instance, one day. Below, we provide details of our model---the process of news generation and spread, users' activity of flagging the news, and selecting news to get expert's labels.

%\noindent
%{\bfseries Generation and spread of news.}
\subsection{News Generation and Spread}
We assume that new news, denoted by the set $X^t$, are generated at the beginning of every epoch $t$ (cf., line~\ref{interaction.generation}).\footnote{For simplicity of presentation, we consider every news generated in the network to be unique. In real-world settings, the same news might be posted by multiple users because of externalities, and it is easy to extend our model to consider this scenario.}
In this paper, we consider a setting where each news has an underlying label (unknown to the algorithm) of being ``fake" ($f$) or ``not fake" ($\bar{f}$). We use random variable $Y^*(x)$ to denote this unknown label for a news $x$ and its realization is given by $y^*(x) \in \{f, \bar{f}\}$. The label $y^*(x)$ can only be acquired if news $x$ is sent to an expert for review who would then provide the true label. We maintain a set of ``active" news $A^t$  (cf., line~\ref{interaction.active}) which consists of all news that have been generated by the end of epoch $t$ but for which expert's label have not been acquired yet. 
%:\mathcal{X} \rightarrow 

\noindent Each news $x$ is associated with a source user who seeded this news, denoted as $o_x$  (cf., line~\ref{interaction.generation}). We track the spread of news in the set $A^t$ via a function $\pi^t\colon A^t \rightarrow 2^U$. For a news $a \in A^t$, the function $\pi^t(a)$ returns the set of users who have seen the news $a$ by the end of epoch $t$. During epoch $t$, let $u^t(a) \subseteq U \setminus \pi^{t-1}(a)$ be the set of additional users (possibly the empty set) to whom news $a \in A^t$ propagates in epoch $t$, hence $\pi^t(a) = \pi^{t-1}(a) \cup u^t(a)$ (cf., line~\ref{interaction.spread}).

\subsection{Users' Activity of Flagging the News}
In epoch $t$, when a news $a \in A^t$ propagates to a new user $u \in u^t(a)$, this user can flag the news to be fake. We denote the set of users who flag news $a$ as fake in epoch $t$ via a set $l^t(a) \subseteq u^t(a)$ (cf., line~\ref{interaction.flag}). Furthermore, the function $\psi^t(a)$ returns the complete set of users who have flagged the news $a$ as fake by the end of epoch $t$.\footnote{Note that as per specification of Protocol~\ref{interaction}, for any news $x$, the source user $o_x$ doesn't participate in flagging $x$.}
%%%%%%%%%%
For any news $x$ and any user $u \in U$, we denote the label user $u$ would assign to $x$ via a random variable $Y_u(x)$. We denote the realization of $Y_u(x)$ as $y_u(x) \in \{f, \bar{f}\}$ where $y_u(x) = f$ signifies that user has flagged the news as fake. In this paper, we consider a simple, yet realistic, probabilistic model of a user's flagging activity as discussed below.

\noindent {\bfseries User abstaining from flagging activity.}
Reflecting the behavior of real-world users, user $u$ might abstain from actively reviewing the news content (and by default, does not flag the news)---we model this happening with a probability $\gamma_u \in [0, 1]$. Intuitively, we can think of $1-\gamma_u$ as the engagement of user $u$ while participating in this crowdsourcing effort to detect fake news: $\gamma_u = 1$ means that the user is not participating at all.
% and $\gamma_u = 0$ means that the user is reviewing the content of every news and flagging it appropriately.
%\item $\alpha_u \in [0, 1]$ denotes the probability the user $u$ would not participate in flagging activity for news $x$;

\noindent {\bfseries User's accuracy in flagging the news.}
With probability $(1 - \gamma_u)$, user $u$ reviews the content of news $x$ and labels the news.
% as fake (via flagging) or not (as default).
We model the accuracy/noise in the user's labels, conditioned on that the user is reviewing the content, as follows:
\begin{itemize}
\item $\alpha_u \in [0, 1]$ denotes the probability that user $u$ would not flag the news as fake, conditioned on that \emph{news $x$ is not fake and the user is reviewing the content}.
\item $\beta_u \in [0, 1]$ denotes the probability that user $u$ would flag the news as fake, conditioned on that \emph{news $x$ is fake and the user is reviewing the content}. 
\end{itemize}

%\equiv \big(P(Y_u(x) = \bar{f} | h^*(x) = \bar{f})\big)
% \equiv \big(P(Y_u(x) = f | h^*(x) = f)\big)

%
%(i) Given that $h^*(x) = \bar{f}$ (news is not fake), $P(Y_u(x) = \bar{f} | h^*(x) = \bar{f})$ is given by  $h^*(x)$ .. $\alpha_u \in [0, 1]$ denotes the probability the user $u$ would not flag the news ($Y_u(a) = \bar{f}$); and (ii) Given that  $Y^*(a) = f$ (news is fake), $\beta_u \in [0, 1]$ denotes the probability the user $u$ would flag the news ($Y_u(a) = f$). 
%

\noindent {\bfseries User's observed activity.}
Putting this together, we can quantify the observed flagging activity of user $u$ for any news $x$ with the following matrix defined by variables $(\userPrior_{u,\bar{f}}, \userPrior_{u,f})$:
\begin{equation*}
\begingroup
\renewcommand*{\arraystretch}{1.3}
\setlength\arraycolsep{1.5pt}
\begin{bmatrix}
    \userPrior_{u,\bar{f}}  & 1 - \userPrior_{u,f} \\
    1 - \userPrior_{u,\bar{f}}  & \userPrior_{u,f} \\
\end{bmatrix}
=\\\gamma_u
\begin{bmatrix}
    1  & 1 \\
    0  & 0 \\
\end{bmatrix}
+
(1 - \gamma_u)
\begin{bmatrix}
    \alpha_u  & 1 - \beta_u \\
    1- \alpha_u  & \beta_u \\
\end{bmatrix}
\endgroup
\end{equation*}
where
%\begin{align}
%\delta_{u,\bar{f}} &\equiv P\big(Y_u(x) = \bar{f} \ |\ Y^*(x) = \bar{f}\big)  \\
%1 - \delta_{u,\bar{f}} &\equiv P\big(Y_u(x) = f \ |\ Y^*(x) = \bar{f}\big)  \\
%\delta_{u,f} &\equiv P\big(Y_u(x) = f \ |\ Y^*(x) = f\big)  \\
%1 - \delta_{u,f} &\equiv P\big(Y_u(x) = \bar{f} \ |\ Y^*(x) = f\big)
%\end{align}
%
\begin{equation*}
\begin{cases}
\userPrior_{u,\bar{f}} &\equiv P\big(Y_u(x) = \bar{f} \mid Y^*(x) = \bar{f}\big)  \\
1 - \userPrior_{u,\bar{f}} &\equiv P\big(Y_u(x) = f \mid Y^*(x) = \bar{f}\big)  \\
\userPrior_{u,f} &\equiv P\big(Y_u(x) = f \mid Y^*(x) = f\big)  \\
1 - \userPrior_{u,f} &\equiv P\big(Y_u(x) = \bar{f} \mid Y^*(x) = f\big)
\end{cases}
\end{equation*}

%\[
%\begin{bmatrix}
%    P\big(Y_u(a) = \bar{f} | Y^*(a) = \bar{f}\big)   & P\big(Y_u(a) = \bar{f} | Y^*(a) = f\big) \\
%    P\big(Y_u(a) = f | Y^*(a) = \bar{f}\big)   & P\big(Y_u(a) = f | Y^*(a) = f\big) \\
%\end{bmatrix}
%\]
%\[
%=\\\gamma_u
%\begin{bmatrix}
%    1  & 1 \\
%    0  & 0 \\
%\end{bmatrix}
%+
%(1 - \gamma_u)
%\begin{bmatrix}
%    \alpha_u  & 1 - \beta_u \\
%    1- \alpha_u  & \beta_u \\
%\end{bmatrix}
%\]
%Hence, we have:
%\begin{align}
%&P\big(Y_u(a) = \bar{f} | Y^*(a) = \bar{f}\big) = \gamma_u + (1 - \gamma_u)\alpha_u \\
%&P\big(Y_u(a) = f | Y^*(a) = \bar{f}\big)  = (1 - \gamma_u) (1- \alpha_u )\\
%&P\big(Y_u(a) = f | Y^*(a) = f\big) = (1 - \gamma_u) \beta_u \\
%&P\big(Y_u(a) = \bar{f} | Y^*(a) = f\big) = \gamma_u + (1 - \gamma_u)(1 - \beta_u)
%\end{align}
%%More specifically, for each user $u \in U$ and any news $a$, we model the flagging activity via three parameters denoted as $(\alpha_u, \beta_u, \gamma_u)$, with the following interpretation:

\noindent The two parameters $(\alpha_u, \beta_u)$ allow us to model users of different types that one might encounter in real-world settings. For instance, 
\begin{itemize}
  \item a user with $(\alpha_u \geq 0.5, \beta_u \leq 0.5)$ can be seen as a ``news lover"  who generally tends to perceive the news as not fake; on the other hand, a user with $(\alpha_u \leq 0.5, \beta_u \geq 0.5)$ can be seen as a ``news hater" who generally tends to be skeptical and flags the news (i.e., label it as fake).
  %\item for a user with $\alpha_u \geq (1 - \beta_u$), the user's label $y_u(x)$ is better than a random coin flip, irrespective of the true label $y^*(x)$; in contrast, $\alpha_u \leq (1 - \beta_u$) signifies the spam behavior of a user.
  \item a user with $(\alpha_u=1,\beta_u=1)$ can be seen as an ``expert'' who always labels correctly; a user with $(\alpha_u=0,\beta_u=0)$ can be seen as a ``spammer'' who always labels incorrectly.
\end{itemize}

\subsection{Selecting News to Get Expert's Label}
At the end of every epoch $t$, we apply an algorithm \generalalgo---on behalf of the network provider---which selects news $S^t \subseteq A^t$  to send to an expert for reviewing and acquiring the true labels $y^*(s) \ \forall s \in S^t$ (cf., line~\ref{interaction.select}). If a news is labeled as fake by the expert (i.e., $y^*(s) = f$), this news is then blocked from the network (cf., line~\ref{interaction.select_block}). At the end of the epoch, the algorithm updates the set of active news as $A^t = A^t \setminus S^t$ (cf., line~\ref{interaction.select_updateactive}). We will develop our algorithm in the next section; below we introduce the formal objective of minimizing the spread of misinformation via fake news in the network.

\subsection{Objective: Minimizing the Spread of Fake News}
Let's begin by quantifying the utility of blocking a news $a \in A^t$ at epoch $t$---it is important to note that, by design, only the fake news are being blocked in the network. Recall that $|\pi^t(a)|$ denotes the number of users who have seen news $a$ by the end of epoch $t$. We introduce $|\pi^\infty(a)|$ to quantify the number of users who would \emph{eventually} see the news $a$ if we let it spread in the network. Then, if a news $a$ is fake, we define the utility of blocking news $a$ at epoch $t$ as $\val^t(a) = |\pi^\infty(a)| - |\pi^t(a)|$, i.e., the utility corresponds to the number of users saved from being exposed to fake news $a$. If an algorithm $\generalalgo$ selects set $S^t$ in epoch $t$, then the total expected utility of the algorithm for $t = 1, \ldots, T$ is given by
\begin{align}
  \textnormal{Util}(T,\generalalgo) = \sum_{t = 1}^{T} \E\Big[\sum_{s \in S^t} \mathbf{1}_{\{y^*(s) = f\}} \val^t(s)\Big] \label{eq.objective}
\end{align}
where the expectation is over the randomness of the spread of news and the randomness in selecting $S^t \ \forall t \in \{1, \ldots, T\}$. 

\noindent In this work, we will assume that the quantity $\val^t(a)$ in Equation~\ref{eq.objective} can be estimated by the algorithm. For instance, this can be done by fitting parameters of an information cascade model on the spread $\pi^t(a)$ seen so far for news $a$, and then simulating the future spread by using the learnt parameters~\cite{du13scalable,zhao15seismic,rizoiu17expecting}.
%Alternatively, we could substitute the quantity  $\val^t(a)$ by a function that is monotonically decreasing w.r.t.\ $t$. %e.g., $val^t(a) \propto t^{-c}$---we will demonstrate this through our experiments.

%sadikov2011correcting

\noindent Given the utility values $\val^t(\cdot)$, we can consider an oracle \oracle that has access to the true labels $y^*(\cdot)$ for all the news and maximizes the objective in Equation~\ref{eq.objective} by simply selecting $k$ fake news with highest utility.
% However, without any knowledge of $y^*(\cdot)$ any algorithm would simply resort to random selection.
In the next section, we develop our algorithm \algo that performs Bayesian inference to compute $y^*(\cdot)$ using the flagging activity of users as well as via learning users' flagging accuracy $\{ \userPrior_{u,\nofake}, \userPrior_{u,\fake} \}_{u \in U}$ from historic data.

\section{Our Methodology}
\label{sec:algorithms}
In this section we present our methodology and our algorithm \algo. We start by describing how news labels can be inferred for the case in which users' parameters are fixed. Next, we consider the case in which users' parameters are unknown and employ a Bayesian approach for inferring news labels and learning users' parameters. Given a prior distributions on the users' parameters and a history of observed data (users' flagging activities and experts' labels obtained), one common approach is to compute a point estimate for the users' parameters (such as MAP) and use that. However, this can lead to suboptimal solutions because of limited exploration towards learning users' parameters. In \algo, we overcome this issue by employing the idea of \emph{posterior sampling}~\cite{thompson1933likelihood,osband2013more}.
%Finally, we derive \algo by replacing the integration with posterior sampling.

\subsection{Inferring News Labels: Fixed Users' Params}
We take a Bayesian approach to deal with unknown labels $y^*(\cdot)$ for maximizing the objective in Equation~\ref{eq.objective}. As a warm-up, we begin with a simpler setting where we fix the users' labeling parameters $(\userPrior_{u, \nofake}, \userPrior_{u, \fake})$ for all users $u \in U$. Let's consider epoch $t$ and news $a \in A^t$ for which we want to infer the true label $y^*(a)$.
%For inference, we have the following information: users $u \in l^t(a)$ have flagged the news, i.e. $y_u(a) = f$; and users $u \in \pi^t(a) \setminus l^t(a)$ have labeled the news as $y_u(a) = \bar{f}$.
Let $\newsPrior$ be the prior that a news is fake; then, we are interested in computing:
\begin{align*}
&P(Y^*(a)=\fake \mid \{\userPrior_{u,\nofake}, \userPrior_{u,\fake}\}_{u \in U}, \newsPrior, \psi^t(a), \pi^t(a))\notag \\
&\propto \newsPrior \cdot \prod_{u \in \psi^t(a)} P\Big(Y_u(a) = f \mid Y^*(a) = f, \userPrior_{u, \fake} \Big) \cdot \notag \\
&\ \ \ \quad\prod_{u \in \pi^t(a) \setminus \psi^t(a)} P\Big(Y_u(a) = \bar{f} \mid Y^*(a) = f, \userPrior_{u, \fake}\Big) \notag\\
&= \newsPrior \cdot \prod_{u \in \psi^t(a)} \userPrior_{u, f} \cdot \prod_{u \in \pi^t(a) \setminus \psi^t(a)}  (1 - \userPrior_{u, f}) %\label{eq:newsPosteriorFixedParameters}
\end{align*}
%&\qquad\propto P\Big(l^t(a), \pi^t(a) \setminus l^t(a)\ |\ Y^*(a) = f\Big) \\
where the last two steps follow from applying Bayes rule and assuming that users' labels are generated independently. Note that both users' parameters $\{\userPrior_{u, \bar{f}}, \userPrior_{u, f}\}_{u \in U}$ affect the posterior probability of a news being fake as the normalization constant depends on both $P(Y^*(a)=\fake \mid \cdot)$ and $P(Y^*(a)=\nofake \mid \cdot)$.

\noindent At every time $t \in \{1, \ldots, T\}$, we can use the inferred posterior probabilities to greedily select $k$ news $S^t \subseteq A^t, |S^t|=k$ that maximize the total expected utility, i.e.,
\begin{align}
   \sum_{s \in S^t} P(Y^*(s)=\fake \mid \cdot) \cdot \val^t(s). \label{eq.objective.bayes}
\end{align}
This greedy selection can be performed optimally by selecting $k$ news with the highest expected utility.
This is implemented in our algorithm \topx, shown in Algorithm~\ref{alg:topx}.

%To show the role of both of the user's parameters $(\userPrior_{u, \bar{f}}, \userPrior_{u, f})$, below we write the normalization constant for this setting:
%\begin{align}
%\newsPrior \cdot \prod_{u \in l^t(a)} \userPrior_{u, f} \cdot \prod_{u \in \pi^t(a) \setminus l^t(a)}  (1 - \userPrior_{u, f})& \notag\\
%+ (1 - \newsPrior) \cdot \prod_{u \in l^t(a)} (1 - \userPrior_{u, \bar{f}}) \cdot \prod_{u \in \pi^t(a) \setminus l^t(a)}  \userPrior_{u, \bar{f}}& \label{knownparams.normfactor}
%\end{align}
%\begin{align}
%&P\Big(\hat{Y}(a) = f \ |\ l^t(a), \pi^t(a) \setminus l^t(a)\Big)\\
%&=\frac{\prod_{u \in l^t(a)} \delta_{u, f} \cdot \prod_{u \in \pi^t(a) \setminus l^t(a)}  (1 - \delta_{u, f})}{\prod_{u \in l^t(a)} \delta_{u, f} \cdot \prod_{u \in \pi^t(a) \setminus l^t(a)}  (1 - \delta_{u, f}) + \prod_{u \in l^t(a)} \delta_{u, f} \cdot \prod_{u \in \pi^t(a) \setminus l^t(a)}  (1 - \delta_{u, f})}
%\end{align}

\subsection{Inferring News Labels: Learning Users' Params}
In our setting, the users' parameters $\{\userPrior_{u, \nofake}, \userPrior_{u, \fake}\}_{u \in U}$ are unknown and need to be learnt over time. 
%In our approach, we jointly infer news labels and users' unknown parameters.

%{\bfseries History and posterior over users' parameters.}
\noindent {\bf Learning about users.} We  assume a prior distribution over the users' parameters  $(\UserPrior_{\nofake}, \UserPrior_{\fake})$ shared among all users. 
%defined by the hyperparameters $\Theta_{f}$ and $\Theta_{\bar{f}}$ as follows:
%\begin{align*}
%\delta_{u, \bar{f}} \sim P(\delta_{u, \bar{f}} \ | \ \Theta_{\bar{f}}); \delta_{u, f} \sim P(\delta_{u, f} \ | \ \Theta_{f})  
%\end{align*}
%
For each user $u \in U$, we maintain the data history in form of the following matrix:
\begin{align*}
\begingroup
\renewcommand*{\arraystretch}{1.3}
\userHistory^t_u = 
\begin{bmatrix}
    d^t_{u, \nofake\mid\nofake}  & d^t_{u, \nofake\mid\fake}  \\
    d^t_{u,\fake\mid\nofake}  &  d^t_{u, \fake\mid\fake}\\
\end{bmatrix}.
\endgroup
\end{align*}
The entries of this matrix are computed from the news for which experts' labels were acquired. For instance, $d^t_{u, \nofake\mid\nofake}$ represents the count of how often the user $u$ labeled a news as not fake and the acquired expert's label was not fake. 

% Let $x$ be any news for which we have acquired the true label from expert $y^*(x)$. If news $x$ spreads to $u$, then we can record the data $\big(y_u(x), y^*(x)\big)$ for this user. More concisely, we represent the data history up to time $t$ with the following matrix giving counts of different types of data points:
%(but excluding  labels obtained from $S^t$ in epoch $t$) as 

\noindent Given $\userHistory^t_u$, we can compute the posterior distribution over the users' parameters using Bayes rules as follows:
\begin{align*}
P(\userPrior_{u,\nofake} \mid \UserPrior_\nofake, \userHistory_u^t) \propto  P(\userHistory_u^t \mid \userPrior_{u,\nofake}) \cdot P(\userPrior_{u,\nofake} \mid \UserPrior_\nofake) \\ 
= (\userPrior_{u, \bar{f}})^{d^t_{u,\nofake \mid \nofake}} \cdot (1 - \userPrior_{u, \bar{f}})^{d^t_{u, \fake \mid \nofake}} \cdot P(\userPrior_{u, \bar{f}} \mid \UserPrior_{\nofake})
\end{align*}
Similarly, one can compute $P(\userPrior_{u,\fake} \mid \cdot)$.

%We will use $\D^t$ to denote the data history of all the users. 

%{\bfseries Reputation of users.}
%Furthermore, to increase robustness and avoid adversarial users favoring the spread of fake news, we introduce some sort of reputation system for users as is often used by real-world crowdsoucring systems (cf., \cite{vesdapunt2014crowdsourcing} for a discussion of Facebook's crowdsourcing system that uses signals from only users deemed credible based on their history). In our work, we would incorporate signals from only those users for which we have high certainty, i.e., variance over their unknown parameters is lower than  threshold $\tau$. 
%%%\cite{resnick2000reputation}: Reputation systems.
%%%\cite{vesdapunt2014crowdsourcing}: Facebooks crowdsourcing system.
%At time $t$, we measure uncertainty over users parameters as follows:
%\begin{align}
%\max\bigg(H\Big(P(\delta_{u, \bar{f}} \ | \ \Theta_{\bar{f}}, D^t_u)\Big), H\Big(P(\delta_{u, f} \ | \ \Theta_{f}, D^t_u)\Big)\bigg) \label{user.variance}
%\end{align}
%
%where function $H(.)$ quantifies some notion of uncertainty in the posterior distribution---we will use $H(.)$ as variance in this paper. We define a set of credible users $U^t \subseteq U$ to be those users $u \in U$ such that:
%\begin{align*}
%\max\bigg(H\Big(P(\delta_{u, \bar{f}} \ | \ \Theta_{\bar{f}}, D^t_u)\Big), H\Big(P(\delta_{u, f} \ | \ \Theta_{f}, D^t_u)\Big)\bigg) \leq \tau
%\end{align*}

%Recall that $tau$ is the variance threshold to decide which users we used to leverage their signals. 
\noindent {\bfseries Inferring labels.}
We can now use the users' parameters posteriors distributions to infer the labels, for instance, by first computing the MAP parameters
\begin{equation*}
  \userPrior^{\textnormal{MAP}}_{u,\nofake} = \arg \max_{\userPrior_{u,\nofake}} P(\userPrior_{u,\nofake} \mid \UserPrior_\nofake, \userHistory_u^t)
\end{equation*} (and $\userPrior^{\textnormal{MAP}}_{u,\fake}$ similarly) and invoking the results from the previous section.\footnote{Note that a fully Bayesian approach for integrating out uncertainty about users' parameters in this case is equivalent to using the mean point estimate of the posterior distribution.}
\noindent Then, at every epoch $t$ we can invoke \topx with a point estimate for the users' parameters to select a set $S^t$ of news. However this approach can perform arbitrarily bad compared to an algorithm that knows the true users' parameters (we refer to this algorithm as \algoStar) as we show in our analysis. 
 The key challenge here is that of actively trading off exploration (selecting news that maximize the value of information towards learning users' parameters) and exploration (selecting news that directly expected utility at a given epoch). This is a fundamental challenge that arises in sequential decision making problems, e.g., in multi-armed bandits~\cite{chapelle2011empirical}, active search~\cite{vanchinathan2015discovering,chen17efficient} and reinforcement learning.

\renewcommand{\algorithmcfname}{Algorithm}
%\begin{algorithm2e}[t!]
%\noindent\begin{minipage}{\textwidth}
%\renewcommand\footnoterule{}                  %% This line should come here.
\begin{algorithm}[!t]
\nl {\bf Input}: 
  \begin{itemize}
    \item {Active news $A^t$; information $\val^t(\cdot)$, $l^t(\cdot), \pi^t(\cdot)$
    \item budget $k$; news prior $\newsPrior$
    \item users' parameters $\{\userPrior_{u,\nofake}, \userPrior_{u,\fake}\}_{u \in U}$}.
  \end{itemize}

%\nl {\bf Initialize}: {$S = \{\}$}.\\
%	\While {$|S| < k$}{
%		\nl $s^* = \argmax_{s \in A^t \setminus S} \lambda \cdot F^t_1(S \cup \{s\}; \tau) + (1 - \lambda) \cdot F^t_2(S \cup \{s\})$\\
%		\nl $S = S \cup \{s^*\}$ \\
%  }
%\nl {\bf Return}: $S$  

\nl Compute $p(a)$ for all $a \in A^t$ as \\$\quad P(Y^*(a)=\fake \mid \{\userPrior_{u,\nofake}, \userPrior_{u,\fake}\}_{u \in U}, \newsPrior, l^t(a), \pi^t(a))$\\[.5mm]

\nl Select\\$\quad S^t = \arg\max_{S \subseteq A^t, |S| \leq k} \sum_{a \in S} p(a) \val^t(a)$\\[.5mm]

\nl {\bf Return}: $S^t$
  \caption{Algorithm \topx}
  \label{alg:topx}
\end{algorithm}
%\end{minipage}
%%%%%%%%%%%%%%%%%%%%%%%%%%%%%%%%%%%%%%%%%%%%%%%%%%%%%%%%%%%%%%%%%%%%%%%%%%%%%%
%%%%%%%%%%%%%%%%%%%%%%%%%%%%%%%%%%%%%%%%%%%%%%%%%%%%%%%%%%%%%%%%%%%%%%%%%%%%%%

\subsection{Our Algorithm \algo}

In this section, we present our algorithm \algo, shown in Algorithm~\ref{algorithm}, that actively trades off between exploration and exploitation by the use of posterior sampling aka Thompson sampling~\cite{thompson1933likelihood,osband2013more}. On every invocation, the algorithm samples the users' parameters from the current users' posterior distributions and invokes \topx with these parameters. Intuitively, we can think of this approach as sampling users' parameters according to the probability they are optimal.
%Although this idea dates back to 1933 there has been a resurge of interest and this approach has recently been applied in many applications involving sequential decision making and has been shown to provide state of the art performance. 

\vspace{2mm}
\noindent {\bf Analysis.} We analyze our algorithms in a simplified variant of Protocol~\ref{interaction}, in particular we make the following simplifications:
\begin{enumerate}
  \item There are $M$ sources $o_1, \ldots, o_M$, each generating news every epoch $t$.
  \item For any news $x$ seeded at epoch $t$, $\val^\tau(x) > 0$ only for $\tau=t$. This means that news $x$ reaches it maximum spread at the next timestep $t+1$, hence the utility of detecting that news drops to $0$.
\end{enumerate}
%
%%%%%%%%%%%%%%%%%%%%%%%%%%%%%%%%%%%%%%%%%%%%%%%%%%%%%%%%%%%%%%%%%%%%%%%%%%%%%%
%%%%%%%%%%%%%%%%%%%%%%%%%%%%%%%%%%%%%%%%%%%%%%%%%%%%%%%%%%%%%%%%%%%%%%%%%%%%%%
\renewcommand{\algorithmcfname}{Algorithm}
%\begin{algorithm2e}[t!]
%\noindent\begin{minipage}{\textwidth}
%\renewcommand\footnoterule{}                  %% This line should come here.
\begin{algorithm}[!t]
\nl {\bf Input}: 
  \begin{itemize}
    %\item {Active news $A^t$; information $\val^t(\cdot)$, $l^t(\cdot), \pi^t(\cdot)$
    %\item budget $k$; news prior $\newsPrior$
    \item user priors $\UserPrior_\fake, \UserPrior_\nofake$; users' histories $\{ \userHistory_u^t \}_{u \in U}$.
  \end{itemize}

%\nl {\bf Initialize}: {$S = \{\}$}.\\
%	\While {$|S| < k$}{
%		\nl $s^* = \argmax_{s \in A^t \setminus S} \lambda \cdot F^t_1(S \cup \{s\}; \tau) + (1 - \lambda) \cdot F^t_2(S \cup \{s\})$\\
%		\nl $S = S \cup \{s^*\}$ \\
%  }
%\nl {\bf Return}: $S$  
\nl Sample\\$\quad\userPrior_{u,\nofake} \sim P(\userPrior_{u,\nofake} \mid \UserPrior_\nofake, \userHistory_u^t)$, $\userPrior_{u,\fake} \sim P(\userPrior_{u,\fake} \mid \UserPrior_\fake, \userHistory_u^t)$\\[.5mm]

\nl $S^t \leftarrow $ Invoke \topx with parameters $\{ \userPrior_{u,\nofake}, \userPrior_{u,\fake} \}_{u \in U}$

%\nl Compute $p(a)$ for all $a \in A^t$ as \\$\quad P(Y^*(a)=\fake \mid \{\userPrior_{u,\fake}, \userPrior_{u,\nofake}\}_{u \in U}, \newsPrior, l^t(a), \pi^t(a))$\\[.5mm]

%\nl Select\\$\quad S^t = \arg\max_{S \subseteq A^t, |S| \leq k} \sum_{a \in S} p(a) \val^t(a)$\\[.5mm]

\nl {\bf Return}: $S^t$
  \caption{Algorithm \algo}
  \label{algorithm}
\end{algorithm}
%\end{minipage}
%%%%%%%%%%%%%%%%%%%%%%%%%%%%%%%%%%%%%%%%%%%%%%%%%%%%%%%%%%%%%%%%%%%%%%%%%%%%%%
%%%%%%%%%%%%%%%%%%%%%%%%%%%%%%%%%%%%%%%%%%%%%%%%%%%%%%%%%%%%%%%%%%%%%%%%%%%%%%
To state our theoretical results, let us introduce the regret of an algorithm \generalalgo as
\begin{equation*}
  \textnormal{Regret}(T,\generalalgo) = \textnormal{Util}(T, \algoStar) - \textnormal{Util}(T, \generalalgo).
\end{equation*}
We can now immediately state our first theoretical result, highlighting the necessity of exploration.
\vspace{3mm}
\begin{proposition}
  Any algorithm \generalalgo using deterministic point estimates for the users' parameters suffers linear regret, i.e.,
  $$\textnormal{Regret}(T, \generalalgo) = \Theta(T).$$
  \label{neg-result}
\end{proposition}
\vspace{-5mm}
\begin{proof}[Proof sketch] The proof follows by considering a simple problem involving two users, where we have perfect knowledge about one user with parameters $(0.5 + \epsilon, 0.5 + \epsilon)$ and the other user either has parameters $(1,1)$ or $(0,0)$ (\emph{expert} or \emph{spamer}). The key idea here is that any algorithm using point estimates can be tricked into always making decisions based on the first user's flagging activities and is never able to learn about the perfect second user.
\end{proof}
The above result is a consequence of insufficient exploration which is overcome by our algorithm \algo, as formalized by the following theorem.
\begin{theorem}
  The expected regret of our algorithm \algo is
    $\E[\textnormal{Regret}(T, \algo)] = \mathcal{O}(C\sqrt{M'T\log(CM'T)})$, where $M'=\binom{M}{k}$ and
  $C$ is a problem dependent parameter. $C$ quantifies the total number of realizations of how $M$ news can spread to $U$ users and how these users label the news.
\end{theorem}
\begin{proof}[Proof sketch] The proof of this theorem follows via interpreting the simplified setting as a reinforcement learning problem. Then, we can apply the generic results for reinforcement learning via posterior sampling of~\citet{osband2013more}. In particular, we map our problem to an MDP with horizon $1$ as follows. The actions in the MDP correspond to selecting $k$ news from the $M$ sources, the reward for selecting a set of news $S$ is given by Equation~\ref{eq.objective.bayes} (evaluated using the true users' parameters).
\end{proof}
\noindent Given that the regret only grows as $\mathcal{O}(\sqrt{T})$ (i.e., sublinear in $T$), this theorem implies that \algo converges to \opt as $T \rightarrow \infty$. However, as a conservative bound on $C$ could be exponential in $|U|$ and $M$, convergence may be slow. Nevertheless, in practice we observe competitive performance of \algo compared to \opt as indicated in our experiments. Hence, \algo overcomes the issues in Proposition~\ref{neg-result}, and actively trades off exploration and exploitation.

\begin{figure*}[t]
  \begin{subfigure}[b]{0.32\textwidth}
    \includegraphics[width=5.5cm]{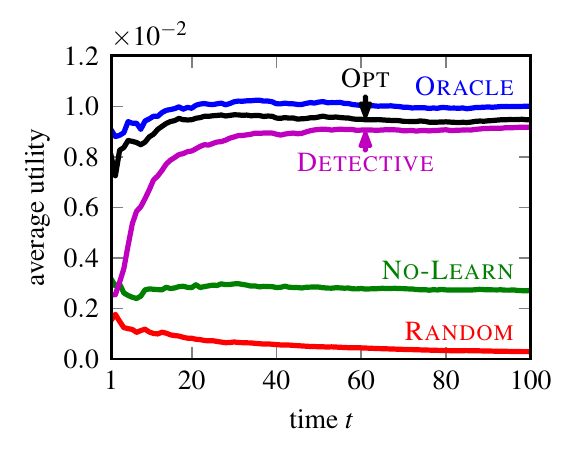}
    \subcaption{Learning about users}
    \label{fig:utility}
  \end{subfigure}%\hspace{3cm}
  \begin{subfigure}[b]{0.32\textwidth}
    \includegraphics[width=5.5cm]{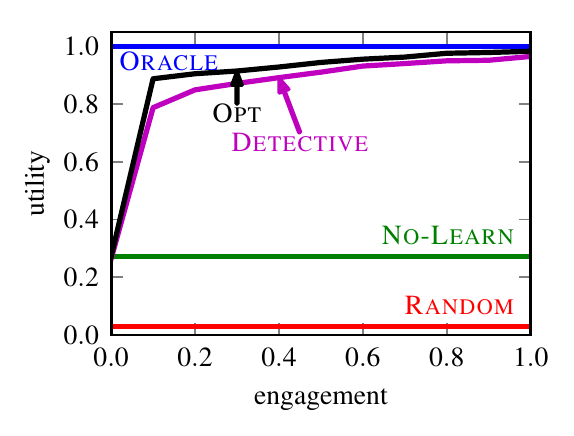}
    \subcaption{Users' engagement in flagging.}
    \label{fig:laziness}
  \end{subfigure} 
  \begin{subfigure}[b]{0.32\textwidth}
    \includegraphics[width=5.5cm]{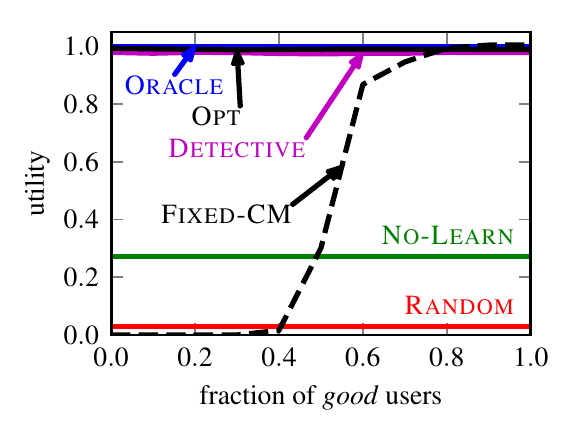}
    \subcaption{Robustness against spammers}
    \label{fig:user-type}
  \end{subfigure}
  \caption{Experimental results.
     \emph{(a)} Learning about users: \algo achieves average utility competitive compared to that of \oracle (which knows the true news labels). The average utility of \algo converges to that of \algoStar as \algo progressively learns the users' parameters.
     \emph{(b)} Users' engagement in flagging: even with low engagement \algo can effectively leverage crowd signals to detect fake news.
     \emph{(c)} Robustness against spammers: \algo is effective even if the majority of users is adversarial, highlighting the importance of learning about users' flagging accuracy for robustly leveraging crowd signals.
     }
  \label{fig:experiments}
\end{figure*}

%%%%%%%%%%%%%%%%%%%%%%%%%%%%%%%%%%%%%%%%%%%%%%%%%%%%%%%%%%%%%
%%%%%%%%%%%%%%%%%%%%%%%%%%%%%%%%%%%%%%%%%%%%%%%%%%%%%%%%%%%%%
\section{Experimental Setup}

%In this section we describe our experimental setup.
% the efficiency of our proposed algorithm for fake news detection in a social network.

\noindent \textbf{Social network graph and news generation.} We consider the \emph{social circles Facebook} graph~\cite{leskovec2012learning}, consisting of 4{,}039 users (nodes) $U$ and 88{,}234 edges, computed from survey data collected by using a Facebook app for identifying social circles. Every user can be the seed of news as described shortly and to every user a probability is assigned with which it (hypothetically) generates fake news in case it seeds news. In particular, $20\%$ of the users generate fake news with probability $0.6$, $40\%$ of the users generate fake news with probability $0.2$ and the remaining $40\%$ of the users generate fake news with probability $0.01$ (the class of a user is assigned randomly). For determining the seeds of news, we partition the users into users $U_n$ which commonly spread news and users $U_r = U \setminus U_n$ which only occasionally spread news. That is, in every iteration of Protocol~\ref{interaction}, we select $M=25$ users for generating news, where users in $U_n$ are selected with probability $\tfrac{0.5}{|U_n|}$ and users in $U_r$ are selected with probability $\tfrac{0.5}{|U_r|}$. Hence, in our experimental setup this corresponds to a prior for seeding fake news of about $20\%$, i.e., $\newsPrior \approx 0.2$.% of all generated news are fake news.

\noindent \textbf{News spreading.} In our experiments, news spread according to an independent cascade model~\cite{kempe2003maximizing}, i.e., the diffusion process of every news is a separate independent cascade with infection probability $0.1 + \mathcal{U}[0, 0.1]$ (fixed when the news is seeded). In every epoch of Protocol~\ref{interaction}, we perform two iterations of the independent cascade models to determine the news spread at the next epoch.
We estimate the number of users who would eventually see news $a$, i.e., $|\pi^\infty(a)|$, by executing the independent cascade models for each news for 600 iterations.

% We randomly select a subset of $U_S \subseteq U$ of 100 users which are  frequent sources of news, i.e., . The other $50 \%$ of the news are generated uniformly at random from the remaining users $U \setminus U_S$. Each user is randomly assigned a probability for generating fake news from $[0.6, 0.01, 0.2, 0.01, 0.2]$, resulting in a prior for seeding fake news of $20\%$.

\noindent\textbf{Users' parameters.} In our experiments we consider three types of users, i.e., \emph{good users} ($\alpha_u = \beta_u = 0.9$), \emph{spammers} ($\alpha_u = \beta_u = 0.1$) and \emph{indifferent users} ($\alpha_u = \beta_u = 0.5$). Unless specified otherwise, each user is randomly assigned to one of these three types.
%The types differ by their parameters $\alpha_u$ and $\beta_u$ (cf., the section on Users' Activity of Flagging the News). 
Also, we set $\gamma_u=0$ unless specified otherwise (note that $1-\gamma_u$ quantifies the engagement of a user).

%\noindent 
%All selected news stop spreading in the network and the stopped spread for selected fake news is credited to the corresponding algorithm.

\noindent \textbf{Algorithms.} We execute Protocol~\ref{interaction} for $T=100$ epochs.
In every epoch of Protocol~\ref{interaction}, the evaluated algorithms select $k=5$ news to be reviewed by an expert.
In our experiments we compare the performance of \algo, \algoStar (unrealistic: \topx invoked with the true users' parameters), \oracle (unrealistic: knows the true news labels). In addition, we consider the following baselines:
\begin{itemize}
  %\item \algo. Our proposed algorithm as described in Section Our Algorithm \algo.
  %\item \algoStar. A fictitious variant of \algo, in which the algorithm is provided the true confusion matrices of the users. The performance of this algorithm is an upper bound for the performance of our algorithm \algo.
  \item \fixedCM. This algorithm leverages users' flags without learning about or distinguishing between users. It uses fixed users parameters $\userPrior_{u,\nofake}=\userPrior_{u,\fake}=0.6$ for invoking \topx.
  %\item \opt. This algorithm selects the $k$ fake news with highest utility, i.e., $S_t = \arg\max_{S \subseteq \{x \mid y^*(x) = f, x \in A^t\}, |S|=k} \sum_{x' \in S} \val^t(x)$.
  \item \nolearn. This algorithm does not learn about users and does not consider any user flags. It greedily selects those news with highest $\val^t(\cdot)$, i.e., $$S^t = \arg\max_{S \subseteq A^t, |S|=k} \sum_{s \in S} \val^t(s).$$
  \item \random. This algorithm selects a random set $S^t \subseteq A^t, |S^t|=k$ of active news for labeling by experts.
\end{itemize}

\section{Experimental Results}

In this section we demonstrate the efficiency of our proposed algorithm for fake news detection in a social network.
All reported utilities are normalized by $\textnormal{Util}(T,\oracle)$ and all results are averaged over $5$ runs.

\noindent \textbf{Learning about users and and exploiting user's flags.} In this experiment we compare the average utility, i.e., $\tfrac{1}{t} \textnormal{Util}(t,\generalalgo)$ (cf., Equation~\ref{eq.objective}), achieved by the different algorithms at epoch $t$ for $t=1, \ldots, T$.
The results are shown in Figure~\ref{fig:utility}. We observe that \algo and \algoStar achieve performance close to that of \oracle. This is impressive, as these algorithms can only use the users' flags and the users' parameters $\{\userPrior_{u,\nofake}, \userPrior_{u,\fake} \}_{u \in U}$ (or their beliefs about the users' parameters in case of \algo) to make their predictions. We also observe that the performance of \algo converges to that of \algoStar as \algo progressively learns the users' parameters. 
%This illustrates that our algorithm \algo effectively learns about users and can also effectively exploit that information.
The algorithms \nolearn and \random achieve clearly inferior performance compare to \algo.

\noindent \textbf{Users' engagement in flagging.} In this experiment, we vary the engagement $1-\gamma_u$ of the users. We report the utilities $\textnormal{Util}(T,\generalalgo)$ in Figure~\ref{fig:laziness}. We observe that with increasing engagement the performance of \algo and \algoStar improves while the performance of the other shown algorithms is not affected by  the increased engagement.
%This is in line with our theoretical understanding of \algo.
Importantly, note that also with a low engagement \algo can effectively leverage crowd signals to detect fake news. 

\noindent \textbf{Robustness against spammers.} In this experiment we consider only two types of users, i.e., good users and spammers. We vary the fraction of good users relative to the total number of users. We report the utilities $\textnormal{Util}(T,\generalalgo)$ achieved by the different algorithms in Figure~\ref{fig:user-type}. We also plot the additional baseline \fixedCM. Observe that the performance of \fixedCM degrades with a decreasing fraction of good users. 
%In contrast, the algorithm \fixedCM requires $40\%$ good users to achieve performance comparable to that of random.
\algo (thanks to learning about users) is effective even if the majority of users is adversarial.
%Its performance increases monotonically with the fraction of good users and saturates slightly below the performance of \algo/\algoStar/\opt. Not that \fixedCM does not converge to the performance of \algoStar as the fixed confusion matrices assumed in \fixedCM does not correspond to the actual confusion matrices of the users.
This highlights the fact that it is crucial to learn about users' flagging accuracy in order to robustly leverage crowd signals.

\section{Conclusions}

In our paper we considered the important problem of leveraging crowd signals for detecting fake news. We demonstrated that any approach that is not learning about users' flagging behaviour is prone to failure in the presence of adversarial/spam users (who want to ``promote'' fake news). We proposed the algorithm \algo that performs Bayesian inference for detecting fake news and jointly learns about users over time. Our experiments demonstrate that \algo is competitive with the fictitious algorithm \opt, which knows the true users' flagging behaviour. Importantly, \algo (thanks to learning about users) is robust in leveraging flags even if a majority of the users is adversarial.
There are some natural extensions for future work. For instance, it would be useful to extend our approach to model and infer the trustworthiness of sources. It would also be important to conduct user studies by deploying our algorithm in a real-world social system.

% based on history of users' flagging activity and expert's labels obtained.
% We demonstrate the effectiveness of our approach via extensive experiments and show the power of leveraging community signals.
% towards leveraging crowd signals for detecting fake news.
%on a publicly available Facebook dataset.

%
%We introduced a novel class of utility functions over sequences of items, thereby extending the expressive power of commonly used submodular set functions. We showed that the naive extensions of classical algorithms fail for our problem of selecting sequences of items of bounded length that maximize the utility. In fact, the search space for the optimal solutions for this new class of functions is exponentially larger than that of the classical subset selection problems. We developed a novel algorithm for sequence selection which takes into account the structural properties of the graph underlying the sequential preferences. Our theoretical analysis provides approximation guarantees for our algorithm w.r.t.\ the intractable optimal solution. We performed experiments on a movie recommendation dataset for the task of recommending sequences of movies to a user. We demonstrate that  several existing baselines can be cast as a special instances of our model and show the effectiveness of our approach in terms of improved accuracy over these baselines.

%%%%%%%%%%%%%%%%%%%%%%%%%%%%%%%%%%%%%%%%%%%%%%%%%%%%%%%%%%%
%%%%%%%%%%%%%%%%%%%%%%%%%%%%%%%%%%%%%%%%%%%%%%%%%%%%%%%%%%% ACKS
\begin{acks}
  This work was supported in part by the Swiss National Science Foundation, and Nano-Tera.ch program as part of the Opensense II project, ERC StG 307036, and a Microsoft Research Faculty Fellowship. Adish Singla acknowledges support by a Facebook Graduate Fellowship.
\end{acks}

%%%%%%%%%%%%%%%%%%%%%%%%%%%%%%%%%%%%%%%%%%%%%%%%%%%%%%%%%%%
%%%%%%%%%%%%%%%%%%%%%%%%%%%%%%%%%%%%%%%%%%%%%%%%%%%%%%%%%%%

\bibliographystyle{ACM-Reference-Format}
%\bibliography{refs}


\end{document}